\documentclass[12pt]{article}

\usepackage[T1]{fontenc}
\usepackage[utf8]{inputenc}
\usepackage{lmodern}
\usepackage[margin=1.12in]{geometry}
\usepackage{microtype}
\usepackage{amsmath,amssymb,amsthm,mathtools}
\usepackage{booktabs,array,enumitem}
\usepackage{graphicx}
\usepackage[authoryear,round]{natbib}
\usepackage{xcolor}
\usepackage{hyperref}
\hypersetup{
 colorlinks=true,
 linkcolor=blue!50!black,
 citecolor=blue!50!black,
 urlcolor=blue!50!black,
 pdftitle={The Fiscal Alibi: Hidden Spending Needs and Government Reputation},
 pdfsubject={Private fiscal needs, reputational cover, and verification},
 pdfkeywords={government reputation, hidden fiscal needs, convex order, auditing, political agency},
 pdfauthor={Georgy Lukyanov and Hengrina Ly}
}
\graphicspath{{figures/}}

\numberwithin{equation}{section}
\setlist{itemsep=0.25em,topsep=0.4em}

\newtheorem{assumption}{Assumption}[section]
\newtheorem{proposition}[assumption]{Proposition}
\newtheorem{theorem}[assumption]{Theorem}
\newtheorem{lemma}[assumption]{Lemma}

\theoremstyle{definition}

\theoremstyle{remark}

\newcommand{\E}{\mathbb E}
\newcommand{\Prob}{\mathop{\rm Pr}}
\newcommand{\dd}{\,d}

\title{\textbf{The Fiscal Alibi}\\
\large Hidden Spending Needs and Government Reputation}
\author{Georgy Lukyanov\thanks{Toulouse School of Economics;
\href{mailto:georgy.lukyanov@tse-fr.eu}{georgy.lukyanov@tse-fr.eu}.}
\and Hengrina Ly\thanks{Université de Rouen Normandie, LERN; ENS Paris-Saclay,
CEPS; \href{mailto:hengrina.ly@univ-rouen.fr}{hengrina.ly@univ-rouen.fr}.}}
\date{September 2026}

\begin{document}
\maketitle

\begin{abstract}
A government may ask for a high tax because it faces a genuine expense. The same demand can also be made by a government that intends to keep the proceeds. We study this ambiguity in a two-period reputation model with privately observed spending needs and an endogenous tax base. The opportunist chooses between outright confiscation and levies that an honest government might impose. We characterize the three possible regimes through a single equilibrium equation. A mean-preserving spread of legitimate needs weakly raises the opportunist's lifetime value, with a strict increase precisely when the spread changes the upper tail relevant for mimicry. We then give a necessary and sufficient curvature condition for this ordering to extend to a general reputational continuation prize. Verification reduces the value of concealment, but also changes how the opportunist extracts. In the benchmark economy, greater auditing weakly lowers current citizen welfare, even before audit costs, while improving the future allocation through better information. The welfare case for verification consequently depends on the balance between these two effects. Examples yield no auditing, an optimum within a regime, and an optimum at the boundary between regimes.
\end{abstract}

\noindent\textbf{Keywords:} government reputation; hidden fiscal needs; taxation; convex order; audits; political agency.\\
\textbf{JEL codes:} D82; D83; H11; H21; H61.

\section{Introduction}

Suppose a government raises taxes to pay for an unexpected expense. Citizens observe the tax demand, but cannot readily establish how much the government needs to spend. A high tax is then consistent with two rather different explanations: the expense may be large, or the government may intend to keep some of the money. This ambiguity matters even if citizens have no difficulty observing how much they pay. What they cannot observe is whether the payment is justified.

The possibility of a genuine expense gives an opportunistic government a reason that citizens cannot immediately reject. It can levy the same amount as an honest government facing a large bill. We refer to this possibility as a fiscal alibi.\footnote{There is no separate report of the spending need in the model. The opportunist chooses a levy, and citizens compare that action with the levies that an honest government could impose. The term \emph{alibi} refers to this inference.} The government has an incentive to preserve the ambiguity because its reputation affects future production and, hence, the amount it can subsequently extract. Its current levy must balance the revenue collected today against the reputation left for tomorrow.

This paper develops a model of that tradeoff. We ask how the distribution of legitimate spending needs affects opportunistic behavior, and whether verifying those needs improves citizen welfare. The two questions are related, but their answers need not have the same sign. Reducing the value of concealment can make a government more willing to confiscate current output. It is therefore necessary to distinguish the rents that a government obtains over its lifetime from the losses that its behavior imposes on citizens at a particular date.

There are two periods and two permanent government types. An honest government privately observes a spending requirement, collects exactly that amount, and finances the required service. An opportunistic government has no valued expenditure. It can collect an ordinary levy that resembles legitimate finance, or confiscate the entire tax base and reveal its type. Citizens supply labor before observing the levy. Their willingness to work depends on the probability that the government will leave the marginal product of their labor with them. Reputation has a continuation value because citizens make another labor-supply decision in the second period, when the honest government respects their property and the opportunist seizes it.

The distinction between an ordinary levy and confiscation is relevant here. An ordinary levy is a fixed assessment at an individual's margin; confiscation takes the individual's whole output. Thus, citizens retain the return to additional work whenever either type uses an ordinary levy. A government that conceals its type can sustain more current production than one that is expected to confiscate immediately. This feature allows us to isolate the effect of reputational cover. We return to proportional taxation in an extension, where the ordinary tax also distorts labor supply and some of the comparative statics change.

Our first result characterizes the equilibrium. An opportunist will not imitate every honest levy. A sufficiently small levy yields too little current revenue, even when it establishes honesty with certainty. Larger levies are attractive, but citizens must then attach a lower probability to honesty. In equilibrium, the loss of continuation value exactly offsets the increase in current revenue throughout the range that the opportunist uses. This indifference condition, together with Bayes' rule, determines the opportunist's likelihood relative to the honest levy distribution.

Integrating that likelihood gives the probability of an ordinary levy and therefore the current tax base. The entire equilibrium can then be recovered from one scalar equation. Depending on the parameters, the opportunist confiscates with certainty, mixes between ordinary levies and confiscation, or uses ordinary levies with certainty. We call the last regime full mimicry. Full mimicry preserves the return to current labor, but does not eliminate extraction: every levy imposed by the opportunistic type is still diverted.

Two qualifications accompany this characterization. First, we assume that even the tax base supported by the honest type alone can finance the largest legitimate need. This keeps the comparison separate from insolvency and political replacement. Second, the uniqueness result concerns an explicitly defined class of regular perfect Bayesian equilibria.\footnote{With continuously distributed levies, individual actions have probability zero. We obtain beliefs from local likelihood ratios rather than allowing an independent choice of belief at each such action. The resulting uniqueness claim is conditional on this restriction; it is not a uniqueness claim for every possible specification of off-path beliefs.} These assumptions allow us to characterize the three regimes and their boundaries without adding a default convention or an announcement game.

The second result concerns fiscal risk. Holding the opportunist's value fixed, its relative likelihood is flat at zero over the range of small needs that it does not imitate. Above that range, the likelihood is increasing and strictly convex in the need. A mean-preserving spread therefore increases the total likelihood of mimicry at the original value. Equilibrium is restored by a weak increase in the opportunist's value. This argument applies across regime boundaries as well as within a regime.

The qualification to this result is economically useful. Additional dispersion has no effect if it only redistributes probability among needs that the opportunist would never imitate. What matters is the upper tail that can justify a sufficiently large levy. We give an exact condition for a strict increase in value in terms of this part of the distribution. The result concerns discounted opportunistic rents; it does not, by itself, establish a corresponding ordering of current diversion or citizen welfare.\footnote{A distribution with uniformly larger needs also raises opportunistic value. That comparison changes the mean. Our risk comparison holds the mean fixed, following the distinction in \citet{RothschildStiglitz1970}.}

We next ask how much this conclusion depends on the continuation economy. For a general increasing reputational prize \(C(\mu)\), where \(\mu\) is the posterior probability of honesty, the relevant curvature condition is \(2C'(\mu)+\mu C''(\mu)\geq0\). This condition is sufficient for the distributional ordering throughout the regular, globally solvent class. It is also necessary for that ordering to hold for every local mean-preserving spread: if the condition fails, one can construct an equilibrium in which a small spread lowers opportunistic value. The terminal labor economy in our benchmark satisfies the condition. The proportional-tax extension shows why it need not hold when reputation must also sustain a distortionary tax base.

Finally, we introduce verification after an ordinary levy has been collected. An audit certifies an honest government and exposes an opportunist. It reduces the continuation reward from successful concealment and eventually makes the opportunist abandon mimicry. However, the adjustment begins before that point. When the opportunist still mimics with certainty, weaker future rewards lead it to collect more through current levies. When it mixes, verification raises the probability of outright confiscation and reduces current production.

These responses imply that, in the benchmark model, current citizen welfare weakly falls with auditing even before administrative costs. This conclusion depends on what the audit does: it reveals information after collection, without recovering the diverted funds or imposing a fine.\footnote{Recovery or a current penalty would introduce an additional deterrent. We omit these instruments to identify the effect of verification through reputation. This is also why the analysis does not amount to the optimal verification contract studied, for example, by \citet{Townsend1979}.} Citizens nevertheless benefit in the second period, when better information allows them to supply more labor under an honest government and avoid supplying it to an exposed opportunist. We derive both components of welfare. Neither dominates in general, and the examples include an optimal audit of zero, an optimum within the full-mimicry regime, and an optimum at the transition to partial mimicry.

\subsection{Related literature}

The reputational mechanism is related to \citet{Phelan2006} and \citet{Lu2013}, in which government behavior and beliefs interact with private economic decisions. More generally, \citet{FudenbergGaoPei2022} and \citet{LuoWolitzky2025} study reputation when actions, observations, and private information shape what behavior can be sustained. Our two-period structure is more restrictive than a recurrent reputation game. Its advantage for the present question is that the distribution of spending needs enters equilibrium through an explicit likelihood function. This makes it possible to compare entire distributions and to identify when the comparison reverses.

The hidden justification for a public expenditure also connects the model to political agency. \citet{RogoffSibert1988}, \citet{CoateMorris1995}, and \citet{BesleySmart2007} examine how information problems affect fiscal choices and political accountability. \citet{GrossmanVanHuyck1988} distinguish excusable contingencies from opportunistic repudiation in sovereign borrowing. In our setting, the ambiguity arises within a tax demand: a levy that covers a large legitimate need is observationally indistinguishable from an opportunistic levy of the same size.

Related questions arise in work on fiscal opacity, budget rules, and volatile resources. \citet{SloofBeetsmaSteinweg2026} study fiscal opacity and debt ceilings, while \citet{Siemroth2024} considers private spending needs and the allocation of funds within organizations. \citet{Cage2009} relates aid volatility to rent extraction under asymmetric information. We focus on the distributional comparison rather than an optimal fiscal rule. The main result identifies which changes in hidden needs raise reputational rents, and the verification analysis shows why reducing those rents need not improve the current allocation.

The information problem differs from that in \citet{AblyatifovLukyanov2026}. There, uncertainty concerns whether an authorized fiscal mandate is delivered. Here, citizens observe the levy but do not know how much legitimate spending was needed in the first place. A large realized need can consequently provide cover for a government with nothing to finance. The distinction is central to the comparative statics, rather than an alternative interpretation of the same signal.

Section \ref{sec:model} sets out the environment and the belief restriction. Section \ref{sec:equilibrium} characterizes equilibrium. Section \ref{sec:risk} studies changes in the need distribution and the role of continuation curvature. Section \ref{sec:audit} derives the effects of verification on behavior and welfare, and Section \ref{sec:numerics} illustrates the results. The appendices contain the proofs, a discussion of solvency, and the proportional-tax extension.

\section{Environment}
\label{sec:model}

\subsection{Government types and spending needs}

There are two periods. The government has a permanent type \(\theta\in\{H,B\}\), and citizens begin with the prior
\begin{equation}
 p=\Prob(\theta=H)\in(0,1).
 \label{eq:prior}
\end{equation}

An honest government, type \(H\), is committed to financing a legitimate public service. In period 1 it privately observes the amount \(G\) required to provide that service, levies \(T=G\), and spends the proceeds accordingly. An opportunistic government, type \(B\), has no valued expenditure and keeps what it collects. Citizens observe the levy, but neither the need nor the contemporaneous use of the money. Section \ref{sec:audit} allows these facts to be verified after collection.

The distribution \(F\) of legitimate needs is common knowledge. It has a continuous density \(f\), strictly positive on a compact interval \([\underline g,\overline g]\). Positive density will allow us to specify beliefs at every levy in the honest government's support. It does not require the opportunist to use that entire support.

There is a unit mass of atomistic citizens, each of whom produces one unit of output per unit of labor. A citizen's per-period utility is
\begin{equation}
 c-\frac{\ell^{1+\gamma}}{1+\gamma},
 \qquad \gamma>0.
 \label{eq:utility}
\end{equation}

An ordinary levy is a fixed per-capita assessment. A citizen can therefore increase consumption by working more without increasing that assessment. Confiscation is different: it takes the citizen's entire output. Citizens choose labor before learning which levy will be imposed.\footnote{The fixed assessment isolates the informational role of the levy. It is not equivalent to a proportional tax. With a proportional tax, an ordinary levy also lowers the marginal return to labor; Appendix \ref{app:rates} examines the resulting change in the argument.}

We allow citizens to borrow against other wealth to meet a fixed assessment, so there is no individual cash-on-hand constraint. Aggregate feasibility nevertheless requires the legitimate expense to be no greater than aggregate output \(L\). The benefit of honest spending is \(B(G)\). For welfare comparisons we use \(B(G)=G\), under which the benefit offsets the honest levy. A different benefit function adds \(p\E[B(G)-G]\) to welfare and does not affect the policy comparisons below.

\subsection{Timing, production, and solvency}

In period 1, Nature draws the government type and, if the type is honest, its spending need. Citizens then choose labor using their prior and the anticipated government strategy. After production, the honest government levies its need. The opportunist chooses either an ordinary levy \(T\in[0,L)\) or confiscation of \(L\). A demand for the whole output is defined to be confiscation, rather than a separate ordinary levy. Citizens observe the action and update their beliefs. In the verification extension, an audit follows an ordinary levy.

There is no new spending need in period 2. Citizens choose labor given the posterior \(\mu\) that the government is honest. The honest government leaves their output with them, whereas the opportunist confiscates it. The expected return to a marginal unit of labor is consequently \(\mu\), and terminal labor is
\begin{equation}
 L_2(\mu)^\gamma=\mu,
 \qquad L_2(\mu)=\mu^{1/\gamma}.
 \label{eq:terminal-labor}
\end{equation}

This output is also the opportunist's terminal seizure prize:
\begin{equation}
 C(\mu)=\mu^{1/\gamma}.
 \label{eq:benchmark-C}
\end{equation}

Both citizens and the government discount terminal payoffs by \(\beta\in(0,1]\). The terminal period thus gives reputation a payoff derived from a labor-supply decision.\footnote{The absence of a second spending shock is an assumption about the continuation economy. It allows us to study the first-period information problem without specifying how a government defaults or is replaced after an unfinanceable need. It should not be read as a stationary model of recurrent public expenditure.}

Let \(m\) denote the probability that the opportunist chooses an ordinary levy. In period 1, citizens retain the marginal return to labor whenever the government is honest or the opportunist uses such a levy. Their common labor choice therefore satisfies
\begin{equation}
 L^\gamma=p+(1-p)m.
 \label{eq:labor}
\end{equation}

The smallest tax base consistent with this condition occurs at \(m=0\). Denote it by
\begin{equation}
 r=p^{1/\gamma}.
 \label{eq:r}
\end{equation}

It is sustained by the possibility that the government is honest, even when an opportunist is expected to confiscate with certainty. We impose the following restriction on legitimate needs.

\begin{assumption}
\label{ass:solvency} The essential supremum of the legitimate need satisfies
\begin{equation}
 0\leq\underline g<\overline g<r<1.
 \label{eq:solvency}
\end{equation}
\end{assumption}

Every honest need can then be financed from the smallest equilibrium base. In particular, an honest levy is strictly smaller than \(L\) and is never confused with the action defined as confiscation. The restriction is stronger than solvency along a particular equilibrium path. Its purpose is to make all three regimes, including the outcome under perfect verification, comparable within one economy. Appendix \ref{app:relax-solvency} explains what is needed when the upper bound on needs exceeds \(r\).

\subsection{Beliefs and equilibrium}

We next specify how citizens interpret an ordinary levy. This requires care because a single levy has probability zero under the continuous honest distribution. Let \(M\) be the opportunist's conditional measure over ordinary levies, with total mass \(m\leq1\), and let \(I_\varepsilon(g)=(g-\varepsilon,g+\varepsilon)\). An assessment is \emph{regular} if it satisfies the following conditions:
\begin{enumerate}[label=(\alph*),leftmargin=2.2em]
\item The honest levy distribution is \(F\). Whenever the denominator is positive, the posterior at an ordinary levy is obtained from the local likelihood ratio
\[
\mu(g)=\lim_{\varepsilon\downarrow0}
\frac{pF(I_\varepsilon(g))}
{pF(I_\varepsilon(g))+(1-p)M(I_\varepsilon(g))}.
\]
\item This limit exists at every \(g\in[\underline g,\overline g]\).
\item Confiscation, any singular component of \(M\), and ordinary levies outside the support of \(F\) receive posterior zero.
\end{enumerate}

The first two conditions make beliefs consistent with the relative probabilities of nearby levies. The third treats actions that the honest type does not generate as evidence of opportunism. We will derive absolute continuity of the opportunist's equilibrium levy measure; it is not imposed on its strategy at the outset.

A regular perfect Bayesian equilibrium is a regular assessment in which citizens and the opportunist act sequentially rationally, while the honest type follows its commitment to levy and spend \(G\). The uniqueness results below concern outcomes in this class.\footnote{The local-likelihood requirement restricts beliefs at null actions. We do not prove that every sequence of fully mixed perturbations selects this requirement, and do not assert uniqueness when beliefs at such actions are left unrestricted.}

\section{Equilibrium: betrayal and mimicry}
\label{sec:equilibrium}

Consider first the opportunist's choice after labor has been supplied. Confiscation yields \(L\) immediately and reveals the government's type, leaving no terminal tax base. An ordinary levy can yield less today while preserving a positive continuation payoff. Denote the opportunist's equilibrium lifetime value by \(V\). On every levy that it uses with positive likelihood, the two components of its return must add up to this value:

\begin{equation}
 g+\beta\mu(g)^{1/\gamma}=V.
 \label{eq:indifference}
\end{equation}
The equilibrium ordinary-levy measure is absolutely continuous with respect to \(F\), as established in the proof below. Write \(h=\dd M/\dd F\) for its relative likelihood. This is a likelihood relative to the honest distribution, rather than a density with respect to the levy itself. Bayes' rule gives

\begin{equation}
 \mu(g)=\frac{p}{p+(1-p)h(g)}.
 \label{eq:bayes-rn}
\end{equation}
Substitution into the indifference condition determines that relative likelihood:

\begin{equation}
 h_V(g)=\frac{p}{1-p}\psi_V(g),
 \qquad
 \psi_V(g)=
 \left[\left(\frac{\beta}{V-g}\right)^\gamma-1\right]_+,
 \label{eq:rn-density}
\end{equation}
Confiscation is always available, so \(V\geq L\geq r>\overline g\). The denominator in \eqref{eq:rn-density} is therefore positive. The positive-part operator has a direct interpretation. If \(g\leq V-\beta\), even a posterior of one cannot make the levy strictly preferable to the equilibrium return. The opportunist assigns it no positive density. We call this the frugal range. At larger levies, indifference requires progressively less favorable beliefs. The resulting posterior is

\begin{equation}
 \mu_V(g)=
 \begin{cases}
 1,&g\leq V-\beta,\\[2pt]
 \left((V-g)/\beta\right)^\gamma,&g>V-\beta,
 \end{cases}
 \qquad g\in[\underline g,\overline g].
 \label{eq:posterior}
\end{equation}
Integrating the relative likelihood gives the probability that the opportunist uses an ordinary levy:

\begin{equation}
 m(V)=\frac{p}{1-p}J_F(V),
 \qquad
 J_F(V)=\int\psi_V(g)\dd F(g).
 \label{eq:J}
\end{equation}
Thus \(J_F(V)\) aggregates the scope for mimicry at a given value. The corresponding ordinary-levy density with respect to \(g\) is \(pf(g)\psi_V(g)/(1-p)\). The function \(\psi_V\) will also be the object used to compare need distributions. Its properties follow from the tradeoff already described: a larger current levy can compensate for a worse reputation, whereas a higher equilibrium value makes each particular levy less attractive.

\begin{lemma}
\label{lem:J} For \(V>\overline g\), \(J_F(V)\) is continuous and non-increasing.  It is strictly decreasing on \((\overline g,\overline g+\beta)\) and equals zero for \(V\geq\overline g+\beta\).  For each fixed \(V\), \(\psi_V(g)\) is non-decreasing and convex in \(g\), flat below \(V-\beta\), and strictly convex above it.
\end{lemma}

It remains to make these choices consistent with labor supply and with probabilities adding up to one. If confiscation occurs with positive probability, the opportunist must be indifferent between it and its ordinary levies, so \(V=L\). If it never confiscates, \(m=1\) and labor reaches its undistorted level \(L=1\), while the value of mimicry must be at least one. These observations give the following classification.

\begin{theorem}
\label{thm:classification} Under Assumption \ref{ass:solvency}, there is a unique regular equilibrium outcome. Let \(r=p^{1/\gamma}\).
\begin{enumerate}[label=(\roman*),leftmargin=2.3em]
\item \textbf{Betrayal only.}  If
 \begin{equation}
  r\geq\overline g+\beta,
  \label{eq:betrayal-condition}
 \end{equation}
the opportunist confiscates surely and
 \begin{equation}
  V=L=r,\qquad m=0.
  \label{eq:betrayal-outcome}
 \end{equation}

\item \textbf{Mimicry and betrayal.}  If
 \begin{equation}
  r<\overline g+\beta,
  \qquad p[1+J_F(1)]<1,
  \label{eq:interior-condition}
 \end{equation}
there is a unique \(V\in(r,1)\) solving
 \begin{equation}
  V^\gamma=p[1+J_F(V)].
  \label{eq:interior-root}
 \end{equation}
The equilibrium probabilities and labor are
 \begin{equation}
  m=\frac{V^\gamma-p}{1-p},
  \qquad e=\frac{1-V^\gamma}{1-p},
  \qquad L=V,
  \label{eq:interior-outcome}
 \end{equation}
where \(e\) is the opportunist's confiscation probability.

\item \textbf{Full mimicry.}  If
 \begin{equation}
  r<\overline g+\beta,
  \qquad p[1+J_F(1)]\geq1,
  \label{eq:full-condition}
 \end{equation}
there is a unique \(V\in[1,\overline g+\beta)\) satisfying
 \begin{equation}
  J_F(V)=\frac{1-p}{p}.
  \label{eq:full-root}
 \end{equation}
The opportunist always mimics, \(m=1\), confiscation is absent, and \(L=1\). Equality in \eqref{eq:full-condition} gives \(V=1\).
\end{enumerate} Across all three regimes, \(V\) is equivalently the unique root of
\begin{equation}
 \boxed{p[1+J_F(V)]=\min\{V^\gamma,1\}.}
 \label{eq:unified-root}
\end{equation}
The two regime transitions are continuous.
\end{theorem}

The scalar equation combines the government's incentives with the citizens' response. Its left-hand side is the marginal return to current labor implied by mimicry. Its right-hand side imposes indifference with confiscation whenever that action is used, and probability adding-up when it is not. The left-hand side is decreasing over the active range; the right-hand side increases up to one. Global solvency ensures that the crossing occurs where all legitimate levies remain feasible.

It is useful to distinguish full mimicry from honest behavior. More mimicry supports a larger tax base because it reduces the frequency of complete confiscation. Once \(m=1\), that base cannot increase further. Yet the opportunist continues to keep every ordinary levy it collects. A large base can therefore coexist with extraction that is concealed in legitimate-looking tax demands.

\section{Spending risk and opportunistic rents}
\label{sec:risk}

\subsection{The global convex-order result}

We now hold the prior, preferences, and continuation economy fixed, and change the distribution of legitimate needs. The distributions may have different compact supports, but each must satisfy Assumption \ref{ass:solvency}. Write \(F_1\preceq_{cx}F_2\) when \(F_2\) is a mean-preserving spread of \(F_1\): the means coincide and \(\E_{F_2}\phi(G)\geq\E_{F_1}\phi(G)\) for every convex function \(\phi\) \citep{RothschildStiglitz1970}.

The comparison can be made in two steps. At the value generated by \(F_1\), convexity of \(\psi_V\) implies that the spread weakly raises \(J_F\). The left-hand side of the equilibrium equation consequently rises at the original solution. Since the difference between the two sides decreases through its crossing, the new equilibrium value cannot be lower. No assumption that the two economies remain in the same regime is needed.

\begin{theorem}
\label{thm:convex-order} Suppose \(F_1\preceq_{cx}F_2\), with each distribution satisfying global solvency.  Then
\begin{equation}
 V(F_2)\geq V(F_1).
 \label{eq:value-order}
\end{equation}
Across all regimes,
\begin{equation}
 V(F_2)>V(F_1)
 \quad\Longleftrightarrow\quad
 J_{F_2}(V(F_1))>J_{F_1}(V(F_1)).
 \label{eq:strictness}
\end{equation}
In the interior regime, a strict value increase raises \(m\) and \(L\) and lowers \(e\). In the full-mimicry regime, \(m=L=1\) already.
\end{theorem}

The strictness condition identifies which spending risks matter. At the original value, the kernel is zero below \(c=V-\beta\). A spread entirely confined to that frugal range changes neither \(J_F\) nor the equilibrium. By contrast, additional risk can increase the opportunities for mimicry when it reaches levies above the cutoff.

To express this distinction directly in terms of tails, consider any fixed \(V>\overline g^\vee\), where \(\overline g^\vee=\max\{\overline g_1,\overline g_2\}\). For \(c<\overline g^\vee\), the difference in mimicry functionals has the following representation:

\begin{align}
 J_{F_2}(V)-J_{F_1}(V)
 &=\frac{\gamma}{\beta}
 \left\{\E_{F_2}(G-c)_+-\E_{F_1}(G-c)_+\right\}
 \notag\\
 &\quad+\int_c^{\overline g^{\vee}}
 \frac{\gamma(\gamma+1)\beta^\gamma}{(V-t)^{\gamma+2}}
 \left\{\E_{F_2}(G-t)_+-\E_{F_1}(G-t)_+\right\}\dd t.
 \label{eq:stop-loss}
\end{align}
The quantity \(\E(G-t)_+\) is the expected part of the need above a threshold \(t\). Convex order weakly increases it at every threshold. Equation \eqref{eq:stop-loss} weights these tail differences only at and above the mimicry cutoff, with strictly positive weights. If the cutoff is at or above both supports, the kernels and their difference are zero.

This representation also explains the limit of the welfare interpretation. The theorem orders the opportunist's discounted return, which includes both current revenue and its future seizure opportunity. It does not separately order current diversion. Nor does the improvement in current production when confiscation becomes less frequent establish a citizen-welfare gain.

A first-order stochastic increase in legitimate needs gives a separate comparison. Since \(\psi_V\) is increasing, such an increase also weakly raises value by the same crossing argument. It should be distinguished from the risk experiment, which keeps expected legitimate spending fixed.

\subsection{A necessary and sufficient curvature condition}

The power form of the terminal prize follows from the benchmark labor economy. To see which part of the risk result depends on that economy, replace the opportunist's terminal prize by a general function \(C\). Keep the first-period fixed levy, labor supply, and global solvency assumptions unchanged. Suppose \(C:[0,1]\to[0,1]\) is continuous, twice continuously differentiable on \((0,1)\), satisfies \(C(0)=0\), \(C(1)=1\), and has \(C'(\mu)>0\) in the interior.

The normalization retains a zero prize for an exposed opportunist and a unit prize at complete trust. Indifference on a mimicked levy now takes the form

\begin{equation}
 g+\beta C(\mu)=V.
 \label{eq:general-indifference}
\end{equation}
Inverting \(C\) gives the posterior needed to support any candidate return. Bayes' rule then converts that posterior into a relative likelihood. For \(x\in(0,1]\), write

\begin{equation}
 K_C(x)=\frac{1}{C^{-1}(x)}-1.
 \label{eq:Kc}
\end{equation}
The analogue of \(\psi_V\), including the range in which the opportunist does not mimic, is

\begin{equation}
 \kappa_{V,C}(g)=
 \begin{cases}
 0,&g\leq V-\beta,\\[2pt]
 K_C((V-g)/\beta),&g>V-\beta.
 \end{cases}
 \label{eq:general-claimed-kernel}
\end{equation}
Its integral is

\begin{equation}
 J_{F,C}(V)=\int \kappa_{V,C}(g)\dd F(g),
 \label{eq:J-general-C}
\end{equation}
The piecewise definition keeps the inverse within its domain. The equilibrium equation remains \eqref{eq:unified-root}, with \(J_{F,C}\) replacing \(J_F\). Monotonicity of \(C\) is enough for the classification and the unique crossing. The distributional comparison requires an additional curvature restriction.

\begin{theorem}
\label{thm:curvature} The likelihood kernel generated by \(C\) is convex in the claimed need for every \(V\) if and only if
\begin{equation}
 \boxed{2C'(\mu)+\mu C''(\mu)\geq0
 \quad\text{for every }\mu\in(0,1).}
 \label{eq:C-condition}
\end{equation}
Under \eqref{eq:C-condition}, Theorem \ref{thm:convex-order} extends to \(C\) in every regular, globally solvent regime.  Conversely, if \eqref{eq:C-condition} fails at some \(\mu_0\), there are parameters and a sufficiently small local mean-preserving spread around an active need for which the opportunist's equilibrium value strictly falls.
\end{theorem}

The necessity statement is about the scope of a general theorem. Failure of \eqref{eq:C-condition} does not imply that every spread lowers rents. It means that one can place an active equilibrium in a region of concavity and find a local spread that does so. The proof constructs this comparison with positive continuous densities and strict global solvency.

The benchmark \(C(\mu)=\mu^{1/\gamma}\) satisfies the condition strictly for every \(\gamma>0\). A continuation economy with proportional taxation can instead have a fiscal-capacity threshold below which no positive financing base exists. Near that threshold, the relevant kernel can be concave. Appendix \ref{app:rates} develops this separate economy and gives explicit reversals. Because its continuation technology is not defined as a normalized positive financing branch over every posterior, the appendix is an illustration of the boundary, rather than a direct application of the theorem without further qualifications.

\section{Verification and citizen welfare}
\label{sec:audit}

\subsection{Global equilibrium effects}

Suppose citizens can verify the fiscal need and the use of funds after an ordinary levy has been collected. The audit takes place with a common probability \(\alpha\in[0,1]\), independently of the levy and the government type. An audited honest government is certified, while an audited opportunist is exposed. Conditional on the levy, the absence of an audit conveys no further information.

Verification changes terminal beliefs. It does not recover current revenue or impose a fine, and citizens anticipate its probability when they choose labor. For an opportunist that mimics, the continuation prize is received only when the levy escapes verification. Thus, for \(\alpha<1\), the equilibrium is obtained by replacing \(\beta\) in the government's incentive condition with

\begin{equation}
 \widetilde\beta=\beta(1-\alpha).
 \label{eq:effective-beta}
\end{equation}
The citizens' discount factor remains \(\beta\). This distinction will matter for welfare: \(\widetilde\beta\) measures the opportunist's effective reward from concealment, not society's valuation of the future.

Under perfect verification, no ordinary levy can preserve the opportunist's reputation. It then prefers confiscation to any ordinary levy below \(L\). The resulting base is \(r\), which still finances every honest need by Assumption \ref{ass:solvency}. Write \(V(\alpha)\), \(m(\alpha)\), and \(L(\alpha)\) for the equilibrium value, mimicry probability, and labor.

\begin{proposition}
\label{prop:audit-path} Suppose the unaudited economy is active, so \(r<\overline g+\beta\), and define
\begin{equation}
 \alpha_D=1-\frac{r-\overline g}{\beta}\in(0,1).
 \label{eq:audit-deterrence}
\end{equation}
Then:
\begin{enumerate}[label=(\roman*),leftmargin=2.3em]
\item \(V(\alpha)\) is continuous and non-increasing, and strictly decreasing while mimicry is active.  The mimicry probability and labor are weakly decreasing, while outright confiscation is weakly increasing.
\item For \(\alpha\geq\alpha_D\), the economy is in betrayal only: \(V=L=r\) and \(m=0\).
\item If the unaudited economy has full mimicry, there is a unique \(\alpha_F\in[0,\alpha_D)\) such that \(V(\alpha_F)=1\).  Full mimicry holds below \(\alpha_F\), interior mimicry and betrayal hold on \((\alpha_F,\alpha_D)\), and betrayal only holds above \(\alpha_D\).  If the unaudited economy is already interior, the first region is absent.
\end{enumerate} Both transitions are continuous.  Current concealed revenue need not be monotone in \(\alpha\).
\end{proposition}

The threshold \(\alpha_D\) is the point at which even the largest honest levy, accompanied by the best possible unaudited reputation, is no better than confiscation. At and beyond that point, an ordinary levy already identifies the honest type. Additional audits provide no information that actions do not reveal.

Before that threshold, the adjustment depends on the initial regime. In full mimicry, the probability of an ordinary levy remains one, so labor remains at one. The opportunist instead changes the distribution of its levies to compensate for a smaller future reward. The welfare calculation below shows that expected current diversion rises on this branch. Once the economy enters partial mimicry, further verification also raises confiscation and reduces labor. A fall in lifetime value therefore need not take the form of less extraction today.

\subsection{Welfare accounting}

To evaluate auditing, we need to account for both periods and for the change in behavior. For \(\alpha<1\), let \(\psi_{V,\alpha}\) be the kernel in \eqref{eq:rn-density} with \(\widetilde\beta\) in place of \(\beta\). Two averages will be useful:

\begin{align}
 D(V,\alpha)&=\E\left[G\psi_{V,\alpha}(G)\right],
 \label{eq:diverted-levy}\\
 Q(V,\alpha)&=\E\left[
 \min\left\{1,\frac{V-G}{\widetilde\beta}\right\}
 \right].
 \label{eq:Q}
\end{align}
The ex-ante amount diverted through ordinary levies is \(pD\). To see the weighting, the opportunist's conditional density is \(pf\psi/(1-p)\), and its ex-ante probability is \(1-p\). The second average, \(Q\), is expected terminal labor following an unaudited honest levy. In betrayal only, an ordinary levy reveals honesty; we set \(D=0\) and \(Q=1\), including at \(\alpha=1\).

Let \({\cal C}(\alpha)\) be the real resource cost of maintaining audit capacity, with \({\cal C}(0)=0\). Citizens bear this cost through a separable lump-sum charge, so it does not alter labor incentives or the feasibility of legitimate spending.\footnote{This is an ex-ante capacity cost, rather than a fee incurred only when a particular levy is investigated. A cost of the latter kind would also depend on the equilibrium probability of an ordinary levy.}

Using the normalization \(B(G)=G\), citizen welfare is

\begin{equation}
 \boxed{
 W(\alpha)=
 \frac{\gamma}{1+\gamma}L(\alpha)^{1+\gamma}
 -pD(V(\alpha),\alpha)
 +\frac{\beta p\gamma}{1+\gamma}
 \left[\alpha+(1-\alpha)Q(V(\alpha),\alpha)\right]
 -{\cal C}(\alpha).}
 \label{eq:welfare}
\end{equation}
The first term is current output retained by citizens, net of labor costs and with honest spending credited at its benefit. The second subtracts ordinary levies kept by the opportunist. Confiscation is already reflected in the probability of retaining output and in equilibrium labor, so it must not be subtracted again.

The third term is discounted terminal surplus. An audited honest government induces the full-information labor choice, while an exposed opportunist induces no labor. Following an unaudited levy, labor is instead based on the posterior associated with that levy. This information effect can be made particularly transparent by writing \(q=\widetilde\beta=\beta(1-\alpha)\) and expressing the terminal component as

\begin{equation}
 W_2(\alpha)=\frac{\beta p\gamma}{1+\gamma}
 -\frac{p\gamma}{1+\gamma}
 \E\left[(G+q-V(q))_+\right].
 \label{eq:terminal-welfare-stoploss}
\end{equation}
Here \(V(q)\) is the same equilibrium value parameterized by the effective continuation weight. The positive-part term measures the loss relative to full information. On every active branch \(V_q<1\). Reducing \(q\) by increasing auditing therefore reduces that loss and raises terminal citizen welfare. This improvement must be compared with the current behavioral response.

\begin{proposition}
\label{prop:audit-welfare} Consider uniform audits under the benchmark continuation technology and suppose the unaudited economy is active.
\begin{enumerate}[label=(\roman*),leftmargin=2.3em]
\item Gross current citizen welfare is weakly decreasing in \(\alpha\) in every regime. It falls strictly in the interior regime and, in full mimicry, whenever the active need distribution is non-degenerate.  It is constant in betrayal only.
\item Gross intertemporal citizen welfare is continuous on \([0,1]\).  If \({\cal C}\) is continuous, net citizen welfare is continuous and an optimal audit exists.  If, additionally, \({\cal C}\) is non-decreasing, some optimum can be selected with \(\alpha^*\leq\alpha_D\); if it is strictly increasing, no optimum exceeds \(\alpha_D\).
\item Even when \({\cal C}\equiv0\), the sign of the intertemporal citizen-welfare effect is not fixed.  The terminal information gain in \eqref{eq:welfare} can outweigh the current loss, but need not do so.
\end{enumerate}
\end{proposition}

The proposition separates a current loss from a future gain. While the opportunist still conceals its type, the prospect of preserving a future base restrains what it takes immediately. Verification weakens that restraint. Citizens gain later because they can condition production on better information about the government. There is no general reason for the second effect to dominate the first.

If the unaudited economy is already in betrayal, both actions reveal type and gross welfare is constant in audit intensity. A non-decreasing capacity cost then makes zero auditing optimal, though it need not be the only optimum if the cost is flat. More generally, auditing beyond \(\alpha_D\) has no allocation or information benefit in this model. The relevant policy choice lies between no verification and the point where actions themselves fully separate the types.

\section{Numerical illustrations}
\label{sec:numerics}

Two examples help distinguish the changes in rents from the welfare case for verification. Both use uniform needs and \(\gamma=1\), for which the mimicry functional can be integrated in closed form. The calculations locate equilibrium and regime boundaries continuously; the grid is used only to display the resulting curves.\footnote{The parameter values illustrate the mechanisms rather than estimate their size in a particular country. The replication files include the closed-form solver and a separate validator that reconstructs the examples by numerical integration.}

For the risk comparison, set \(p=\beta=0.5\). Compare \(G\sim U[0.15,0.35]\) with \(G\sim U[0.05,0.45]\). Expected spending is \(0.25\) in both economies, and the second distribution is a mean-preserving spread of the first. Both satisfy global solvency. Without audits, both equilibria involve partial mimicry, and the spread changes value and the mimicry probability as follows:

\begin{equation}
 V: 0.64573\longrightarrow0.66736,
 \qquad
 m:0.29146\longrightarrow0.33472.
 \label{eq:numerical-risk}
\end{equation}
The conditional probability of confiscation falls from \(0.70854\) to \(0.66528\). The larger value therefore coincides with a larger current base. The left panel of Figure \ref{fig:audit} shows how auditing removes the additional cover. In the wider-need economy, even costless verification lowers intertemporal citizen welfare over the range where it changes behavior. Zero auditing is optimal.

For the policy example, set \(p=0.7\), \(\beta=1\), and \(G\sim U[0.4,0.6]\). The unaudited economy has full mimicry and \(V=1.20476\). With no audit cost, welfare is maximized at \(\alpha_F=0.29534\), where full mimicry gives way to partial mimicry. A cost of \({\cal C}(\alpha)=0.1\alpha^2\) leaves the optimum at that boundary. With \({\cal C}(\alpha)=0.2\alpha^2\), the optimum moves to \(\alpha^*=0.24888\), strictly within the full-mimicry regime. The right panel displays all three cases.

The positive audit optima do not eliminate concealment. At the lower optimum, the opportunist still uses an ordinary levy with certainty. At the regime boundary, confiscation is just becoming competitive. The welfare gain comes from better terminal information, despite the increase in current diversion.

\begin{figure}[t]
\centering
\includegraphics[width=0.96\textwidth]{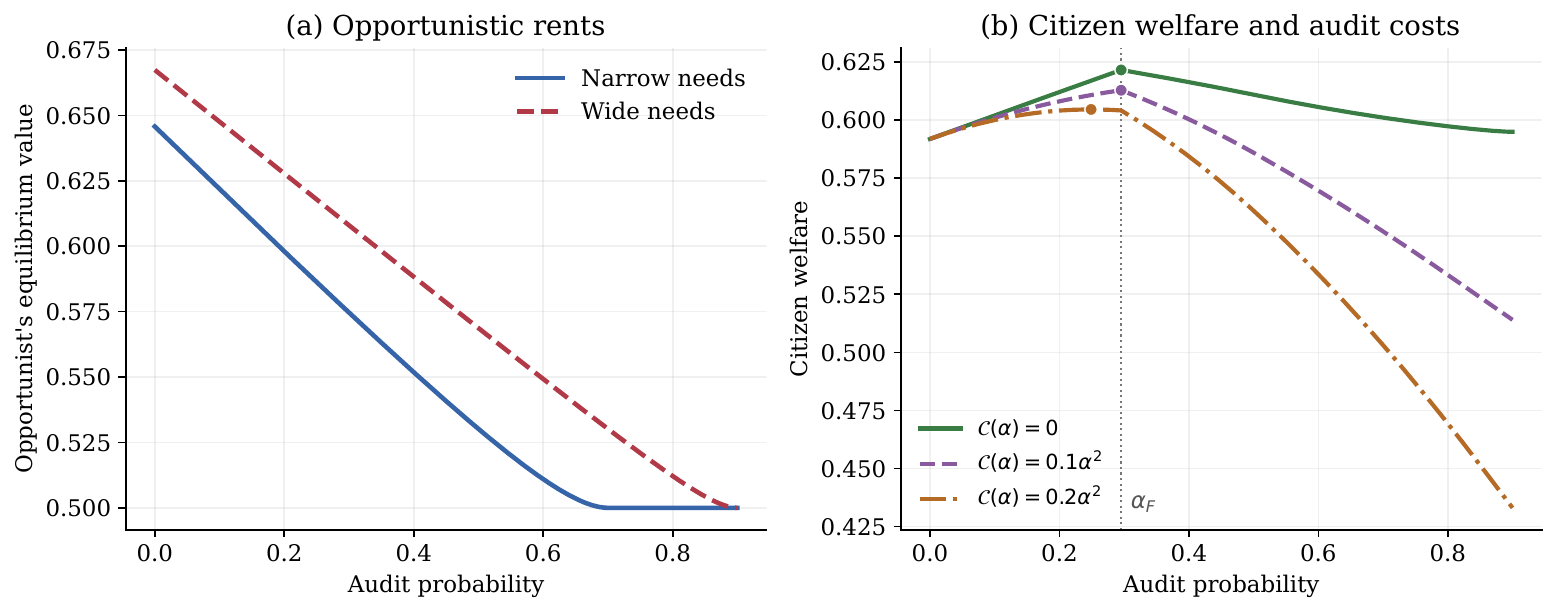}
\caption{Spending risk and verification. Left: opportunistic value for the narrow and wide uniform distributions, with \(p=\beta=0.5\). Right: citizen welfare for \(p=0.7\), \(\beta=1\), and \(G\sim U[0.4,0.6]\), under three audit-cost functions. Dots mark the continuously located optima; the vertical line marks the full-to-partial-mimicry boundary \(\alpha_F\).}

\label{fig:audit}
\end{figure}

\section{Conclusion}
\label{sec:conclusion}

The spending needs of an honest government affect the opportunities available to an opportunist. When citizens see the levy but cannot verify its justification, a government with nothing to finance can imitate the demands of a government facing a large expense. The distribution of those expenses then becomes a determinant of reputational rents.

The model identifies the relevant distributional comparison. A mean-preserving spread weakly increases opportunistic value because the likelihood required to support mimicry is convex in the legitimate need. The increase is strict only when the spread changes the upper tail used for mimicry. Risk among needs that are too small to be imitated has no effect. The general continuation-prize result states the curvature condition behind this argument and shows how a local reversal can arise when it fails.

Verification reduces the future reward from concealment. Its current consequences are less favorable for citizens: the opportunist demands more through ordinary levies or becomes more willing to confiscate, and current welfare weakly falls. Better information improves the subsequent production decision. Whether auditing is worthwhile depends on that future gain relative to the current loss and the resource cost of verification.

These conclusions rely on the stated fiscal instruments and on a continuation economy in which reputation determines a future seizure opportunity. The proportional-tax extension demonstrates that changing the marginal tax incentive can change the distributional and audit responses. Recurrent spending needs, fiscal caps, or stabilization insurance would require further choices about financing, default, and political replacement. The present analysis leaves those choices open while identifying the reputational effect that such models would have to take into account.

\appendix

\section{Proofs}
\label{app:proofs}

\begin{proof}[Proof of Lemma \ref{lem:J}] We first establish the shape of the likelihood kernel, and then integrate it. In the active range \(g>V-\beta\), its derivatives are
\begin{equation}
 \frac{\partial\psi_V(g)}{\partial g}
 =\gamma\beta^\gamma(V-g)^{-\gamma-1}>0,
 \qquad
 \frac{\partial^2\psi_V(g)}{\partial g^2}
 =\gamma(\gamma+1)\beta^\gamma(V-g)^{-\gamma-2}>0.
 \label{eq:A-kernel-derivatives}
\end{equation}

In the frugal range the kernel is zero. The two pieces agree at the cutoff, where the slope increases from zero on the left to \(\gamma/\beta\) on the right. The joined function is therefore non-decreasing and convex, with strict convexity on the active range.

For each need, the kernel is non-increasing in \(V\). On a compact subset of \(V>\overline g\), the denominator stays bounded away from zero, so dominated convergence gives continuity of its integral. If \(V<\overline g+\beta\), there is positive probability above \(V-\beta\), where the decrease in \(V\) is strict. The integral is consequently strictly decreasing there. If \(V\geq\overline g+\beta\), no need is active and the integral is zero.
\end{proof}

\begin{proof}[Proof of Theorem \ref{thm:classification}] The proof has three steps. We first identify the form of any regular equilibrium strategy. We then solve the consistency condition for its value. Finally, we verify that the resulting strategy leaves no profitable deviation.

Under the belief restriction, a singular ordinary action receives posterior zero. If it is in the honest support, its payoff is at most \(\overline g<r\leq L\); if it is outside that support, its payoff is \(T<L\). Confiscation strictly dominates in either case. Ordinary actions outside the honest support are dominated for the same reason. Thus the equilibrium measure satisfies \(M\ll F\). Since confiscation is always available,
\begin{equation}
 V\geq L\geq r>\overline g.
 \label{eq:A-value-order}
\end{equation}

At common Lebesgue points of the densities, the local-likelihood rule gives \eqref{eq:bayes-rn}. Positive mimicry density on \(g\leq V-\beta\) would give a posterior below one and a payoff strictly below \(V\). Hence the density is zero almost everywhere on that range. Conversely, a positive-\(F\)-measure set with zero mimicry density above \(V-\beta\) would contain a common Lebesgue point with posterior one. At that point, the opportunist could obtain \(g+\beta>V\). Such a hole in the active range cannot occur in equilibrium. Indifference therefore implies \(h=h_V\) almost everywhere, and
\begin{equation}
 m=\frac{p}{1-p}J_F(V),
 \qquad
 L^\gamma=p+(1-p)m=p[1+J_F(V)].
 \label{eq:A-mass-labor}
\end{equation}

If confiscation is used, \(m<1\) implies \(L<1\), and indifference with confiscation gives \(V=L\). This yields \eqref{eq:interior-root}. If confiscation is unused, the ordinary probabilities must add up to one, giving \(m=L=1\) and \eqref{eq:full-root}. The value must then satisfy \(V\geq1\) to deter confiscation. Together these conditions give \eqref{eq:unified-root}.

Suppose first that \(r\geq\overline g+\beta\). Even the largest honest levy with the best possible reputation pays no more than confiscation at the base \(r\). Betrayal only is therefore an equilibrium, with the outcome in \eqref{eq:betrayal-outcome}. No active equilibrium can have \(V\geq r\) in this case, since its active set would be empty.

Now suppose \(r<\overline g+\beta\). Define the difference between the two sides of the equilibrium equation on \([r,\overline g+\beta]\):
\begin{equation}
 D(V)=p[1+J_F(V)]-\min\{V^\gamma,1\}.
 \label{eq:A-D}
\end{equation}

At \(r\), the active set has positive probability and \(D(r)=pJ_F(r)>0\). At the upper endpoint \(J_F=0\), while \(\overline g+\beta>r\), so \(D<0\). Lemma \ref{lem:J} makes this difference strictly decreasing through its crossing. There is consequently one root. Its position below, at, or above one is determined by \(p[1+J_F(1)]-1\). This gives the two active regimes and the probabilities stated in the theorem.

To verify the candidate, use the continuous density \(h_Vf\) and its local-likelihood posterior. Every mimicked levy yields \(V\). Every unmimicked honest levy has \(g\leq V-\beta\) and hence yields at most \(V\), even at posterior one. An off-support ordinary levy gives posterior zero and pays less than \(L\leq V\); confiscation pays \(L\leq V\). All deviation constraints are satisfied. Honest levies are feasible because \(L\geq r>\overline g\). At the betrayal boundary \(J_F(r)=0\), and at the full-mimicry boundary \(V=1\). Substitution into the probability and labor formulas establishes continuity at both transitions.
\end{proof}

\begin{proof}[Proof of Theorem \ref{thm:convex-order}] Global solvency ensures that each equilibrium value exceeds both support bounds. On this common domain, convexity of the kernel and the definition of convex order imply
\begin{equation}
 J_{F_2}(V)=\E_{F_2}\psi_V(G)
 \geq\E_{F_1}\psi_V(G)=J_{F_1}(V).
 \label{eq:A-J-order}
\end{equation}

Evaluate this inequality at \(V(F_1)\). The difference between the left- and right-hand sides of the equilibrium equation for \(F_2\) is non-negative there. Since that difference crosses zero only once and decreases through the crossing, its root is weakly larger. It is strictly larger precisely when the inequality is strict at \(V(F_1)\). This establishes both the ordering and its exact strictness condition, including comparisons across regimes. The labor and probability comparisons follow from \eqref{eq:interior-outcome} and the saturation \(m=L=1\) under full mimicry.

For the tail representation, take \(V>\overline g^\vee\) and \(c=V-\beta<\overline g^\vee\). The slope jump at \(c\), together with the second derivative on the active range, gives
\begin{equation}
 \psi_V(g)=\frac{\gamma}{\beta}(g-c)_+
 +\int_c^{\overline g^{\vee}}
 \frac{\gamma(\gamma+1)\beta^\gamma}{(V-t)^{\gamma+2}}
 (g-t)_+\dd t,
 \qquad c=V-\beta,
 \label{eq:A-convex-representation}
\end{equation}

Taking expectations under each distribution and subtracting yields \eqref{eq:stop-loss}. If \(c\geq\overline g^\vee\), both kernels vanish on the relevant supports, so their difference is zero.
\end{proof}

\begin{proof}[Proof of Theorem \ref{thm:curvature}] We first derive the curvature condition and check that the active kernel joins its flat part correctly. Write \(x=C(\mu)\). Differentiating the inverse in \(K_C(x)=1/\mu-1\) gives
\begin{equation}
 K_C''(x)=
 \frac{2C'(\mu)+\mu C''(\mu)}
 {\mu^3[C'(\mu)]^3}.
 \label{eq:A-K-curvature}
\end{equation}

The denominator is positive. Since \(x=(V-g)/\beta\) is affine in \(g\), the numerator also determines curvature with respect to the need. To check the join at \(g=V-\beta\), let \(A(\mu)=\mu^2C'(\mu)\). Then
\begin{equation}
 A'(\mu)=\mu[2C'(\mu)+\mu C''(\mu)],
 \qquad
 \frac{\partial K_C((V-g)/\beta)}{\partial g}
 =\frac{1}{\beta A(\mu)}.
 \label{eq:A-general-slope}
\end{equation}

Under \eqref{eq:C-condition}, \(A\) is non-decreasing in \(\mu\). As the need rises, its supporting posterior falls, so the slope of the active kernel rises. At the cutoff, its one-sided limit is non-negative and therefore no smaller than the zero slope of the frugal range. The value is continuous there. The condition is thus sufficient for convexity of the entire kernel, not only its active piece. Necessity follows from \eqref{eq:A-K-curvature} at any interior posterior where the condition fails. Under convexity, the crossing argument in Theorem \ref{thm:convex-order} applies to \(J_{F,C}\).

For the converse, suppose the condition fails at \(\mu_0\). By continuity, the active kernel is strictly concave in a neighborhood of the corresponding need. We construct an interior equilibrium in that neighborhood. Choose \(V\in(0,1)\), set \(p=\mu_0V^\gamma\), and thus obtain \(r=\mu_0^{1/\gamma}V\). By taking \(V\) sufficiently small, we can choose \(\beta\in((V-r)/C(\mu_0),V/C(\mu_0))\) while maintaining \(\beta\leq1\). Set
\begin{equation}
 g_0=V-\beta C(\mu_0)\in(0,r).
 \label{eq:A-converse-g0}
\end{equation}

With a point mass at \(g_0\), the integral is \(J_{F,C}(V)=1/\mu_0-1\), and the equilibrium equation becomes \(p[1+J_{F,C}(V)]=V^\gamma\). The candidate is interior and strictly solvent.

It remains to implement the comparison within the regular density class. Approximate the point mass by a smooth positive density \(F_1^\varepsilon\) on a sufficiently small common support around \(g_0\). Continuity of the unique crossing implies \(V_1^\varepsilon\to V\). For small \(\varepsilon\), the kernel at this equilibrium remains strictly concave near the central need. Spread a small central mass symmetrically within that neighborhood, leaving a common positive background density. The resulting distribution \(F_2^\varepsilon\) is a mean-preserving spread of \(F_1^\varepsilon\), retains global solvency, and satisfies \(J_{F_2^\varepsilon,C}(V_1^\varepsilon)<J_{F_1^\varepsilon,C}(V_1^\varepsilon)\). Its unique crossing is therefore at a strictly lower value. Smallness of the perturbation preserves the interior regime.

For the benchmark prize, direct differentiation gives
\begin{equation}
 2C'(\mu)+\mu C''(\mu)
 =\frac{1}{\gamma}\left(1+\frac{1}{\gamma}\right)
 \mu^{1/\gamma-1}>0.
 \label{eq:A-power-curvature}
\end{equation}
\end{proof}

\begin{proof}[Proof of Proposition \ref{prop:audit-path}] For an imperfect audit, the relevant government incentive is obtained by substituting \(\widetilde\beta=\beta(1-\alpha)\). Holding \(V\) fixed, a higher \(\widetilde\beta\) increases the kernel and its integral, strictly whenever some needs are active. The unique-crossing equation consequently gives a value that increases with \(\widetilde\beta\) and decreases with auditing. The labor and probability comparisons follow from the formulas for the interior and full-mimicry regimes.

The betrayal threshold is reached when \(r\geq\overline g+\widetilde\beta\). Solving this condition for \(\alpha\) gives \eqref{eq:audit-deterrence}. If the initial equilibrium has full mimicry, the expression \(p[1+J_{F,\widetilde\beta}(1)]\) starts at or above one and falls continuously below one before \(\alpha_D\). Strict monotonicity while the active set is nonempty gives a unique crossing \(\alpha_F\). It may equal zero when the unaudited economy is already at the boundary. If the initial equilibrium is partial, this first transition is absent.

At the transitions, the values are respectively \(1\) and \(r\), so both the value and the allocation formulas agree on either side. Perfect verification gives betrayal directly. Finally, \(D(V,\alpha)\) changes both because the kernel depends on the effective weight and because equilibrium value adjusts. There is no general monotonicity of current concealed revenue across the entire audit path; the full-mimicry and betrayal portions already exhibit different responses.
\end{proof}

\begin{proof}[Proof of Proposition \ref{prop:audit-welfare}] We begin with the accounting in \eqref{eq:welfare}. Expected retained period-1 output is \([p+(1-p)m]L=L^{1+\gamma}\). Subtracting labor cost gives \(\gamma L^{1+\gamma}/(1+\gamma)\). Honest spending and its levy offset under \(B(G)=G\); opportunistic ordinary levies have ex-ante measure \(p\psi\,\dd F\), giving the deduction \(pD\).

At an unaudited ordinary levy \(g\), multiply expected terminal surplus by the probability density of observing that levy. Bayes' rule reduces the product to \(p\,\dd F(g)\,\gamma\mu(g)^{1/\gamma}/(1+\gamma)\). With an audit, the full-information surplus \(\gamma/(1+\gamma)\) is obtained only under the honest type. Integrating and discounting gives the terminal term in \eqref{eq:welfare}.

To compare behavior along the audit path, let \(q=\widetilde\beta\). Define the following moments over the active need set:
\begin{equation}
 A_j=\int (V-g)^j\left(\frac{q}{V-g}\right)^\gamma\dd F(g).
 \label{eq:A-Aj}
\end{equation}

In partial mimicry, differentiate the equilibrium equation \eqref{eq:interior-root} and substitute into the derivative of current welfare. The result is
\begin{equation}
 \frac{\dd W_1}{\dd q}
 =\frac{\gamma p}{q}
 \left[A_1-
 \frac{pA_0^2}{V^{\gamma-1}+pA_{-1}}
 \right]>0.
 \label{eq:A-current-interior}
\end{equation}

Cauchy--Schwarz gives \(A_0^2\leq A_1A_{-1}\), and the denominator is strictly larger than \(pA_{-1}\). The displayed derivative is therefore positive. Since auditing lowers \(q\), current welfare strictly falls. Differentiating the same equilibrium equation also yields
\begin{equation}
 V_q=\frac{pA_0}{q(V^{\gamma-1}+pA_{-1})}<1,
 \label{eq:A-Vq-interior}
\end{equation}

The active set satisfies \(V-g<q\), implying \(qA_{-1}>A_0\) and the strict upper bound on \(V_q\).

In full mimicry, labor stays at one and \(J_F(V)=(1-p)/p\). Differentiation of this adding-up condition and of \(D\) gives
\begin{equation}
 \frac{\dd D}{\dd q}
 =\frac{\gamma}{q}
 \left(\frac{A_0^2}{A_{-1}}-A_1\right)\leq0.
 \label{eq:A-current-full}
\end{equation}

Current welfare is a constant minus \(pD\). Cauchy--Schwarz therefore implies that it falls weakly with auditing, and strictly if active needs are non-degenerate. On this branch \(V_q=A_0/(qA_{-1})<1\). In betrayal only, current strategies and welfare do not change with the audit probability.

The identity \(1-Q=q^{-1}\E[(G+q-V)_+]\) gives \eqref{eq:terminal-welfare-stoploss}. On an active branch, \(V_q<1\) means that the expected positive part rises with \(q\). Thus terminal welfare rises with auditing, strictly while mimicry remains active. This establishes the opposite signs of the two gross welfare components.

At \(V=1\) and \(V=r\), the equilibrium formulas and welfare components are continuous. A continuous capacity cost therefore makes net welfare continuous on the compact interval \([0,1]\), ensuring existence of an optimum. For \(\alpha\geq\alpha_D\), ordinary levies reveal honesty and confiscation reveals opportunism. Additional audits change neither allocation nor information. A non-decreasing cost allows selection of an optimum at or below \(\alpha_D\), and a strictly increasing cost excludes any larger optimum. Finally, the two costless examples in Section \ref{sec:numerics} establish that the intertemporal welfare response can have either sign.
\end{proof}

\section{Relaxing global solvency}
\label{app:relax-solvency}

The main solvency assumption gives a common feasible domain for every audit probability, including perfect verification. For a fixed imperfect audit, one can allow \(\overline g\geq r\), provided \(\overline g<1\) is retained. The latter bound is necessary for strict financing feasibility because even full mimicry supplies no more than one unit of output.

A sufficient replacement condition is divergence of the mimicry integral at the upper support bound. With \(\beta\) replaced by the positive effective weight when auditing is present, require

\begin{equation}
 \lim_{V\downarrow\overline g}J_F(V)=+\infty.
 \label{eq:A-upper-tail}
\end{equation}
If the density is bounded away from zero near \(\overline g\), the condition holds for \(\gamma\geq1\). More generally, if \(f(g)\asymp(\overline g-g)^\kappa\), divergence occurs when \(\kappa\leq\gamma-1\).

The crossing argument can then start at \(V\downarrow\overline g\), where its left-hand side diverges, instead of at \(r\). Since \(\overline g<1\), it produces a strictly solvent active root: either partial mimicry with \(L=V>\overline g\), or full mimicry with \(L=1>\overline g\). When \(\overline g\geq r\) and the audit is imperfect, betrayal only is not a separate regime under this divergence condition.

If the limit is finite, denote it by \(J_b\). For \(r<\overline g<1\), a strictly solvent root exists if and only if

\begin{equation}
 p(1+J_b)>\overline g^\gamma.
 \label{eq:A-finite-tail}
\end{equation}
This is the condition for a positive crossing gap at the lower endpoint \(V=\overline g\). If it fails, the argument does not deliver a strictly solvent equilibrium of the specified form. Equality is a boundary case with no financing slack and requires separate treatment; a strict failure cannot support a root with every honest need below labor.

Even when the divergent-tail condition holds for every imperfect audit, perfect verification is different. It removes the continuation reward from mimicry altogether. If \(\overline g>r\), the resulting betrayal base cannot finance some honest needs, and a default or outside-finance rule is needed. At \(\overline g=r\), every need is only weakly financeable, and the top levy meets the action defined as confiscation. The strict inequality in Assumption \ref{ass:solvency} avoids both issues.

\section{Proportional taxation and fiscal capacity}
\label{app:rates}

This appendix changes the tax instrument as well as the continuation economy. Ordinary finance is now proportional to individual output, and the honest government also faces spending needs in the terminal period. The purpose is to show why the benchmark comparisons need not extend to that setting. We characterize a feasible upper financing branch and give counterexamples on it; we do not assert a global equilibrium classification for the rate economy.

Suppose terminal reputation is \(z\), aggregate labor is \(X\), and the honest government finances a need with mean \(\mu_G>0\) using the proportional rate \(G/X\). The opportunist confiscates output. If every honest realization is feasible, the expected marginal retention rate is \(z(1-\mu_G/X)\). With labor curvature \(\gamma\geq1\), the symmetric labor condition is

\begin{equation}
 X^{\gamma+1}=z(X-\mu_G).
 \label{eq:A-rate-labor}
\end{equation}
A positive financing base requires \(X>\mu_G\). Expressing the posterior as \(z=X^{\gamma+1}/(X-\mu_G)\) shows that its minimum occurs at
\begin{equation}
 X_* =\frac{\gamma+1}{\gamma}\mu_G,
 \qquad
 \underline z=\frac{(\gamma+1)^{\gamma+1}}{\gamma^\gamma}\mu_G^\gamma.
 \label{eq:A-rate-threshold}
\end{equation}
If \(\underline z<1\), each \(z\in(\underline z,1]\) has two positive financing roots, with a double root at \(\underline z\) and none below it. If \(\underline z>1\), no posterior supports positive full financing; if \(\underline z=1\), only the double root at complete trust remains. These statements presume that every realized honest need can be financed, since otherwise the marginal tax wedge must incorporate the treatment of default.

We use the upper root. It is locally stable under the labor-response iteration \(X_{\mathrm{next}}=[z(1-\mu_G/X)]^{1/\gamma}\): at a fixed point the derivative is \(\mu_G/[\gamma(X-\mu_G)]\), which is below one on the upper branch and above one on the lower branch. This is a stability statement for that specified adjustment, rather than an additional equilibrium-selection theorem. For \(\gamma=1\), the threshold is \(\underline z=4\mu_G\) and the upper solution is

\begin{equation}
 X_+(z)=\frac{z+\sqrt{z^2-4z\mu_G}}{2}.
 \label{eq:A-rate-upper}
\end{equation}

Let \(q=\beta(1-\alpha)>0\), and let \(X_1\) be the upper solution at posterior one. At an active ordinary levy, indifference requires the continuation base \(x=(V-g)/q\). Upper-branch inversion is valid only for \(x\in[X_*,X_1]\). A need with \(x\geq X_1\) is not mimicked and has zero likelihood; a candidate with \(x<X_*\) cannot be supported on the selected branch. In particular, a continuous equilibrium supported on that branch requires \(\overline g\leq V-qX_*\). Within the valid range, Bayes' rule gives

\begin{equation}
 z=P(x)=\frac{x^{\gamma+1}}{x-\mu_G},
 \qquad
 \phi(x)=\frac{1}{P(x)}-1
 =\frac{x-\mu_G}{x^{\gamma+1}}-1.
 \label{eq:A-rate-kernel}
\end{equation}
Although the likelihood is increasing in the active need, it need not be convex. Differentiation with respect to \(g\) gives

\begin{equation}
 \frac{\partial^2\phi}{\partial g^2}
 =\frac{(\gamma+1)[\gamma x-(\gamma+2)\mu_G]}
 {q^2x^{\gamma+3}}.
 \label{eq:A-rate-curvature}
\end{equation}
Near \(X_*\), the numerator is negative. The part of the kernel associated with the smallest viable financing bases is therefore concave. A second relevant object is the retained-output contribution \(\kappa_V(g)=(V-g)\phi((V-g)/q)\) in the first-period labor condition. Its curvature is

\begin{equation}
 \kappa_V''(g)=\frac{\gamma}{q}
 \frac{(\gamma-1)x-(\gamma+1)\mu_G}{x^{\gamma+2}}.
 \label{eq:A-rate-rent-curvature}
\end{equation}
For unit curvature, this expression is strictly negative throughout the active range. Thus neither the likelihood comparison nor the labor condition can be signed by importing convexity from the fixed-levy economy.

To obtain the first-period equation, consider an equilibrium in which the opportunist mixes with confiscation, so \(V=L\). Under an honest government, the expected retained output is \(V-\mu_G\). Under an opportunistic ordinary levy \(g\), retained output is \(V-g\), and the relative likelihood is \(\phi((V-g)/q)\). Combining this marginal labor condition with Bayes' rule gives

\begin{align}
 V^{\gamma+1}
 &=p\left[V-\mu_G+
   \int (V-g)\phi\!\left(\frac{V-g}{q}\right)dF(g)\right],
 \label{eq:A-rate-fixed-point}\\
 m&=\frac{p}{1-p}\int
   \phi\!\left(\frac{V-g}{q}\right)dF(g),
 \label{eq:A-rate-mimicry}
\end{align}
Here \(m\) is the opportunist's conditional probability of an ordinary levy. The integrals apply on the valid upper-branch range, with zero likelihood at unmimicked needs. These equations are the partial-mimicry closure, not a formula for a full-mimicry regime in which value could exceed labor.

Consider the parameters

\begin{equation}
 \gamma=1,\qquad \mu_G=0.2,\qquad q=0.8,\qquad p=0.75.
 \label{eq:A-rate-example}
\end{equation}
For a deterministic need \(G=0.2\), the admissible upper-branch mixing solution is \(V_0=0.6000000\) and \(m_0=0.6000000\). Now replace that need by the equal-probability spread \(G\in\{0.16,0.24\}\). On the same admissible branch, the scalar equation becomes

\begin{equation}
 V^2=0.75\left[0.8-
 0.128\frac{V-0.2}{(V-0.2)^2-0.0016}\right]
 \label{eq:A-rate-example-root}
\end{equation}
Its unique admissible solution is
\begin{equation}
 V_1=0.5957515,
 \qquad m_1=0.5987586.
 \label{eq:A-rate-reversal}
\end{equation}
The two continuation posteriors are \(0.8607\) and \(0.8082\), both above the financing threshold \(4\mu_G=0.8\). All honest needs are feasible. The spread therefore lowers both opportunistic value and mimicry while remaining on the selected branch. Convolving both need distributions with the same sufficiently small continuous mean-zero noise preserves their convex-order ranking and, by continuity, the strict reversal. The displayed discrete example is not the source of the reversal.

Auditing can also have a different composition effect under proportional taxation. At a regular decreasing crossing, a higher audit rate locally reduces value, but the mimicry probability need not fall. For example, set \(\gamma=1\), \(\mu_G=0.2\), \(p=0.7\), \(F=\delta_{0.2}\), and \(\beta=0.3\). With no audit, the admissible solution is \((V,m)=(0.3665074,0.3557898)\). At \(\alpha=1/15\), the effective weight is \(q=0.28\), and \((V,m)=(0.3514951,0.3851001)\). Even the probability of undetected mimicry increases: \((1-\alpha)m=0.3594267>0.3557898\).

The last comparison reduces rents and current labor, but raises mimicry and lowers confiscation. It illustrates the additional force introduced by the ordinary marginal tax wedge. Both the risk theorem and the audit-composition result in the main text consequently belong to the specified fixed-levy economy; extending them to proportional taxation requires further restrictions.

\end{document}